\documentclass[a4paper]{article}

\usepackage{amsmath}
\usepackage{amsfonts}
\usepackage{amssymb}         
\usepackage{amsopn}
\usepackage{theorem}          
\usepackage{latexsym}
\usepackage{graphicx}        
\usepackage{mathrsfs}           
\usepackage{psfrag}
\usepackage{algorithmic}
\usepackage{algorithm}
\usepackage{color}
\usepackage{verbatim}
\usepackage{url}

\newtheorem{result}{\ }[section]
\theoremstyle{changebreak}                
\newtheorem{thm}[result]{Theorem}

\newtheorem{lem}[result]{Lemma}
\newtheorem{cor}[result]{Corollary}
\newtheorem{prop}[result]{Proposition}

\newenvironment{proof}
 {{\sl Proof.}\hspace*{1 ex}}%
 {{\nopagebreak\hspace*{\fill}$\Box$\par\vspace{12pt}}}

\begin{document}

\thispagestyle{empty}
\begin{center}
{\LARGE Polynomial cases of the Discretizable Molecular Distance
  Geometry Problem}
\par\bigskip
{\sc Leo Liberti${}^1$, Carlile Lavor${}^2$, Beno\^{\i}t Masson${}^3$, Antonio Mucherino${}^4$} 
\par\bigskip
\begin{minipage}{15cm}
\begin{flushleft}
{\small
\begin{itemize}
\item[${}^1$] {\it LIX, \'Ecole Polytechnique, 91128 Palaiseau,
  France} \\  Email:\url{liberti@lix.polytechnique.fr}
\item[${}^2$] {\it Dept.~of Applied Maths (IME-UNICAMP), State
  Univ.~of Campinas, 13081-970, Campinas - SP, Brazil} \\
  Email: \url{clavor@ime.unicamp.br}
\item[${}^3$] {\it IRISA, INRIA, Campus de Beaulieu, 35042 Rennes,
    France} \\
  Email: \url{benoit.masson@inria.fr}
\item[${}^4$] {\it CERFACS, Toulouse, France} \\
  Email: \url{antonio.mucherino@cerfacs.fr}
\end{itemize}
}
\end{flushleft}
\end{minipage}
\par \medskip \today
\end{center}
\par \bigskip

\begin{abstract}
  An important application of distance geometry to biochemistry
  studies the embeddings of the vertices of a weighted graph in the
  three-dimensional Euclidean space such that the edge weights are
  equal to the Euclidean distances between corresponding point pairs.
  When the graph represents the backbone of a protein, one can exploit
  the natural vertex order to show that the search space for feasible
  embeddings is discrete. The corresponding decision problem can be
  solved using a binary tree based search procedure which is
  exponential in the worst case. We discuss assumptions
  that bound the search tree width to a polynomial size. 
\end{abstract}

\section{Introduction}
\label{intro}
We study the following decision problem \cite{dmdgp}:
\begin{quote} 
{\sc Discretizable Molecular Distance Geometry Problem}
(DMDGP). Given a simple undirected weighted graph $G=(V,E,d)$ where
$d:E\to\mathbb{R}_+$, $V$ is ordered so that $V=[n]=\{1,\ldots,n\}$,
and the following assumptions hold:
\begin{enumerate}
\item for all $v>3$ and $u\in V$ with $1\le v-u\le 3$,
 $\{u,v\}\in E$ ({\sc Discretization})
\item for all $v>3$, $E$ contains all edges $\{u,w\}$ with $u\not=w\in
  U_v=\{u\in V\;|\;1\le v-u\le 3\}$, and the distances $d_{uw}$ with
  $u\not=w\in U_v$ obey the strict simplex inequalities
  \cite{blumenthal} ({\sc Strict Simplex Inequalities}),
\end{enumerate}
and given an embedding $x':[3]\to\mathbb{R}^3$, is there an embedding
$x:V\to\mathbb{R}^3$ extending $x'$, such that
\begin{equation}
  \forall \{u,v\}\in E \quad \|x_u-x_v\|=d_{uv} \; ? \label{mdgpeq}
\end{equation}
\end{quote}
Note that the strict simplex inequalities in $\mathbb{R}^3$ reduce to
the strict triangular inequalities
$d_{v-3,v-1}<d_{v-3,v-2}+d_{v-2,v-1}$. An embedding $x$ extends an
embedding $x'$ if $x'$ is a restriction of $x$; an embedding is
feasible if it satisfies \eqref{mdgpeq}. We also consider the
following problem variants:
\begin{itemize}
\item DMDGP${}_K$, i.e.~the family of decision problems (parametrized
  by the positive integer $K$) obtained by replacing each symbol `3'
  in the DMDGP definition by the symbol `$K$';
\item the ${}^{\mbox{\sf\scriptsize K}}$DMDGP, where $K$ is given as
  part of the input (rather than being a fixed constant as in the
  DMDGP${}_K$).
\end{itemize}
We remark that DMDGP=DMDGP${}_3$. Other related problems also exist in
the literature, such as the {\sc Discretizable Distance Geometry
  Problem} (DDGP) \cite{ddgp}, where the {\sc Discretization} axiom is
relaxed to require that each vertex $v>K$ has at least $K$ adjacent
predecessors. The original results in this paper, however, only refer
to the DMDGP and its variants. 

The {\sc Discretization} axiom guarantees that the locus of the points
embedding $v$ in $\mathbb{R}^3$ is the intersection of the three
spheres centered at $v-3,v-2,v-1$ with radii
$d_{v-3,v},d_{v-2,v},d_{v-1,v}$. If this intersection is non-empty,
then it contains two points apart from a set of Lebesgue measure 0
where it may contain either one point or infinitely many. The role of
the {\sc Strict Simplex Inequalities} axiom is to prevent the latter
case of infinitely many points. As such we might actually dispense
with this axiom altogether and simply discuss results that occur with
probability 1. We remark that if the intersection of the three spheres
is empty, then the instance is a NO one. The {\sc Discretization}
axiom allows the solution of DMDGP instances using a recursive
algorithm called Branch-and-Prune (BP) \cite{lln5}: at level $v$, the
search is branched according to the (at most two) possible positions
for $v$. The BP generates a (partial) binary search tree of height
$n$, each full branch of which represents a feasible embedding for the
given graph.

The DMDGP and its variants are related to the {\sc Molecular Distance
  Geometry Problem} (MDGP), which asks to find an embedding in
$\mathbb{R}^3$ of a given weighted undirected graph. We denote the
generalization of the MDGP to embeddings in $\mathbb{R}^K$ where $K$
is part of the input by {\sc Distance Geometry Problem} (DGP), and the
variants with fixed $K$ by DGP${}_K$. The MDGP is a good model for
determining the structure of molecules given a set of inter-atomic
distances \cite{mdgpsurvey,jogomdgp}. Such distances can usually be
found using Nuclear Magnetic Resonance (NMR) experiments
\cite{schlick}, a technique which allows the detection of inter-atomic
distances below 5{\AA}. The DGP has applications in wireless sensor
networks \cite{eren04} and graph drawing. In general, the MDGP and DGP
implicitly require a search in a continuous Euclidean space
\cite{mdgpsurvey}.

The DMDGP is a model for protein backbones. For any atom $v\in V$, the
distances $d_{v-1,v}$ and $d_{v-2,v-1}$ are known because they refer
to covalent bonds.  Furthermore, the angle between $v-2$, $v-1$ and
$v$ is known because it is adjacent to two covalent bonds, which
implies that $d_{v-2,v}$ is also known by triangular geometry. In
general, the distance $d_{v-3,v}$ is smaller than 5{\AA} and can
therefore be assumed to be known by NMR experiments; in practice,
there are ways to find atomic orders which ensure that $d_{v-3,v}$ is
known \cite{iwcp10}. There is currently no known protein with
$d_{v-3,v-1}$ being {\it exactly equal} to $d_{v-3,v-2}+d_{v-2,v-1}$
\cite{lln5}.

The rest of this paper is organized as follows. In Sect.~\ref{s:bp} we
describe the BP algorithm. In Sect.~\ref{s:compl} we discuss
complexity issues. Sect.~\ref{s:bwidth} describes some polynomial
DMDGP subclasses. We make several important contributions: an {\bf
  NP}-hardness proof for the ${}^{\mbox{\sf\scriptsize K}}$DMDGP and
the DMDGP${}_K$ (for $K>2$), a new proof that the number of feasible
embeddings of DMDGP instances is a power of two, and some practically
relevant polynomial cases of the DMDGP.

\section{The BP algorithm}
\label{s:bp}
For all $v\in V$ we let $N(v)=\{u\in V\;|\;\{u,v\}\in E\}$ be the set
of vertices {\it adjacent} to $v$. An embedding of a subgraph of $G$
is called a {\it partial embedding} of $G$. We denote by $X$ the set
of embeddings (modulo congruences) solving a DMDGP${}_K$ (or
${}^{\mbox{\sf\scriptsize K}}$DMDGP) instance.

The BP algorithm exploits the edges guaranteed by the {\sc
  Discretization} axiom in order to search a discrete set: vertex $v$
can be placed in at most two possible positions (the intersection of
$K$ spheres in $\mathbb{R}^K$). Each is tested in turn and the
procedure called recursively for each feasible positions. The BP
exploits all other edges in the graph in order to prune some branches:
a position might be feasible with respect to the distances to the $K$
immediate predecessors $v-1,\ldots,v-K$, but not necessarily with
distances to other adjacent predecessors.

For a partial embedding $\bar{x}$ of $G$ and $\{u,v\}\in E$ let
$S^{\bar{x}}_{uv}$ be the sphere centered at $x_u$ with radius
$d_{uv}$.
\begin{algorithm}[!htp]
\begin{algorithmic}[1]
\REQUIRE A vtx.~$v\in V\smallsetminus [K]$, a partial
embedding $\bar{x}=(x_1,\ldots,x_{v-1})$, a set $X$.  
\STATE $P=\bigcap\limits_{u\in N(v)\atop u<v} S_{uv}^{\bar{x}}$; 
\STATE $\forall p\in P\; (\;(x\leftarrow(\bar{x},p))$; {\bf if}
  ($v=n$) $X\leftarrow X\cup\{x\}$ {\bf else} {\sc BP}($v+1$,
  $x$, $X$)$\;)$.
\end{algorithmic}
\caption{{\sc BP}($v$, $\bar{x}$, $X$)} 
\label{alg:bp}
\end{algorithm}
The BP algorithm, used for solving the DMDGP and its variants, is {\sc
  BP}($K+1$, $x'$, $\emptyset$) (see Alg.~\ref{alg:bp}), where $x'$ is
the initial embedding of the first $K$ vertices mentioned in the DMDGP
definition. By the DMDGP axioms, $|P|\le 2$. At termination, $X$
contains all embeddings (modulo congruences) extending $x'$
\cite{lln5,dmdgp}. Embeddings $x\in X$ can be represented by sequences
$\chi(x)\in\{-1,1\}^n$ with: (i) $\chi(x)_i=1$ for all $i\le K$; (ii)
for all $i>K$, $\chi(x)_i=-1$ if $ax_i < a_0$ and $\chi(x)_i=1$ if
$ax_i\ge a_0$, where $ax=a_0$ is the equation of the hyperplane
through $x_{i-K},\ldots,x_{i-1}$. For an embedding $x\in X$, $\chi(x)$
is the {\it chirality} of $x$ \cite{CH88} (the formal definition of
chirality actually states $\chi(x)_0=0$ if $ax_i=a_0$, but since this
event has probability 0, we do not consider it here).

The BP (Alg.~\ref{alg:bp}) can be run to termination to find all
possible embeddings of $G$, or stopped after the first leaf node at
level $n$ is reached, in order to find just one embedding of $G$. In
the last few years we have conceived and described several BP variants
targeting different problems \cite{dmdgpejor}, including, very
recently, problems with interval-type uncertainties on some of the
distance values \cite{iwcp10}. Compared to continuous search
algorithms (e.g.~\cite{morewu}), the performance of the BP algorithm
is impressive from the point of view of both efficiency and
reliability. The BP algorithm, moreover, is currently the only method
able to find all embeddings for a given protein backbone.

\section{Complexity}
\label{s:compl}
Any class of YES instances where each vertex $v$ only has distances to
the $K$ immediate predecessors provides a full BP binary search tree
(after level $K$), and therefore shows that the BP is an
exponential-time algorithm in the worst case. One remarkable feature
of the computational experiments conducted on our BP implementation
\cite{mdjeep} on protein instances is that the exponential-time
behaviour of the BP algorithm was never noticed empirically. When we
were able to embed protein backbones of ten thousand atoms in just
over 13 seconds of CPU time (on a single core) \cite{aiccsa10}, we
started to suspect that protein instances might have some special
properties ensuring that the BP ran in polynomial time. Specifically,
using the particular structure of the protein graph, we argue in
Sect.~\ref{s:bwidth} that it is reasonable to expect that the BP will
yield a search tree of bounded width.

Restricting $d$ to only take integer values, the DGP${}_1$ is {\bf
  NP}-complete by reduction from {\sc Subset-Sum}, the DGP${}_K$ is
(strongly) {\bf NP}-hard by reduction from 3-SAT, and the DGP is
(strongly) {\bf NP}-hard by induction on $K$ \cite{saxe79}. Only the
DGP${}_1$ is {\bf NP}-complete because if $d$ is integer then the
YES-certificate $x$ (the embedding) can be chosen to have integer
values. It is currently not known whether there is a polynomial length
encoding of the algebraic numbers that can be used to show that DGP is
in {\bf NP}. 

The DMDGP is {\bf NP}-hard by reduction from {\sc Subset-Sum} (Thm.~3
in \cite{dmdgp}). We generalize that proof to the DMDGP${}_K$.
Intuitively, we exploit the fact that a subset sum instance
$a_1,\ldots,a_N$ with solution $s_1,\ldots,s_N\in\{-1,1\}$ has
$\sum_{\ell\le N} s_\ell a_\ell = 0$ (the zero-sum property) to
construct a DMDGP instance with $KN+1$ points, where the zero-th point
is at the origin and the $\ell$-th set of $K$ successive points is
associated to $a_\ell$; the $j$-th point in the $\ell$-th set adds
$s_\ell a_\ell$ to its $j$-th coordinate, so that the last point is
again the origin (all coordinates satisfy the {\sc Subset-Sum}'s
zero-sum property).
\begin{thm}
The DMDGP${}_K$ is {\bf NP}-hard for all $K\ge 2$. \label{thmK}
\end{thm}
\begin{proof}
  Let $a=(a_1,\ldots,a_N)$ be an instance of {\sc Subset-Sum} consisting
  of positive integers, and define an instance of DMDGP${}_K$ where
  $V=\{0,\ldots,KN\}$, $E$ includes $\{i,i+j\}$ for all
  $j\in\{1,\ldots,K\}$ and $i\in\{0,\ldots,KN-j\}$, and:
\begin{eqnarray}
  \forall i\in\{0,\ldots,KN-1\} \quad d_{i,i+1} &=& a_{\lfloor
     i/K\rfloor}  \label{eqj1} \\
  \forall j\in\{2,\ldots,K\}, i\in\{0,\ldots,KN-j\} \quad d_{i,i+j} &=&
     \sqrt{\sum_{\ell=1}^{j} d_{i+\ell-1,i+\ell}^2} \label{eqjell} \\
d_{0,KN} &=& 0. \label{eqjlast}
\end{eqnarray}
Let $s\in\{-1,1\}^N$ be a solution of the {\sc Subset-Sum} instance
$a$. We let $x_0=0$ and for all $i=K(\ell-1)+j>0$ we let
$x_i=x_{i-1}+s_\ell a_\ell e_j$, where $e_j$ is the vector with a one
in component $j$ and zero elsewhere. Because $\sum_{\ell\le N} s_\ell
a_\ell=0$, if $s$ solves the {\sc Subset-Sum} instance $a$ then, by
inspection, $x$ solves the corresponding DMDGP instance
\eqref{eqj1}-\eqref{eqjlast}.  Conversely, let $x$ be an embedding
that solves \eqref{eqj1}-\eqref{eqjlast}, where we assume without loss
of generality that $x_0=0$. Then \eqref{eqjell} ensures that the line
through $x_i,x_{i-1}$ is orthogonal to the line through
$x_{i-1},x_{i-2}$ for all $i>1$, and again we assume without loss of
generality that, for all $j\in\{1,\ldots,K\}$, the lines through
$x_{j-1},x_j$ are parallel to the $i$-th coordinate axis. Now consider
the chirality $\chi$ of $x$: because all distance segments are
orthogonal, for each $j\le K$ the $j$-th coordinate is given by
$x_{KN,j}=\sum\limits_{i\bmod K=j} \chi_i a_{\lfloor i/K\rfloor}$.
Since $d_{0,KN}=0$, for all $j\le K$ we have $0=x_{KN,j}=\sum_{\ell\le
  N} \chi_{K(\ell-1)+j} a_\ell$, which implies that, for all $j\le K$,
$s^j=(\chi_{K(\ell-1)+j}\;|\;1\le\ell\le N)$ is a solution for the
{\sc Subset-Sum} instance $a$. 
\end{proof}

\begin{cor}
  The ${}^{\mbox{\sf\scriptsize K}}$DMDGP is {\bf NP}-hard.
\end{cor}
\begin{proof}
  Every specific instance of the ${}^{\mbox{\sf\scriptsize K}}$DMDGP
  specifies a fixed value for $K$ and hence belongs to the
  DMDGP${}_K$. Hence the result follows by inclusion. 
\end{proof}

\section{BP search trees with bounded width}
\label{s:bwidth}

We partition $E$ into the sets $E_D = \{\{u,v\} \;|\; |v-u|\le K\}$
and $E_P = E\smallsetminus E_D$. We call $E_D$ the {\it discretization
  edges} and $E_P$ the {\it pruning edges}. Discretization edges
guarantee that a DGP instance is in the ${}^{\mbox{\sf\scriptsize
    K}}$DMDGP. Pruning edges are used to reduce the BP search space by
pruning its tree.  In practice, pruning edges might make the set $P$
in Alg.~\ref{alg:bp} have cardinality 0 or 1 instead of 2. We assume
$G$ is a YES instance of the ${}^{\mbox{\sf\scriptsize K}}$DMDGP.

\subsection{The discretization group}
Let $G_D=(V,E_D,d)$ and $X_D$ be the set of embeddings of $G_D$; since
$G_D$ has no pruning edges, the BP search tree for $G_D$ is a full
binary tree and $|X_D|=2^{n-K}$.  The discretization edges arrange the
embeddings so that, at level $\ell$, there are $2^{\ell-K}$ possible
embeddings $x_v$ for the vertex $v$ with rank $\ell$.  We assume that
$|P|=2$ at each level $v$ of the BP tree, an event which, in absence
of pruning edges, happens with probability 1 --- thus many results in
this section are stated with probability 1. Let $x_v,x'_v$ the
possible embeddings of $v$ at level $v$ of the tree; then by
elementary spherical geometry considerations, $x'_v$ is the reflection
of $x_v$ through the hyperplane defined by $x_{v-K},\ldots,x_{v-1}$.
Denote this reflection by $R_x^v$.
\begin{thm}[Cor.~4.5 and Thm.~4.8 in \cite{powerof2-tr}]
\label{pow2lemma}
With probability 1, for all $v>K$ and $u<v-K$ there is a set $H^{uv}$,
with $|H^{uv}|=2^{v-u-K}$, of real positive values such that for each
$x\in X$ we have $\|x_v-x_u\|\in H^{uv}$. Furthermore, $\forall x\in
X\;\|x_v-x_u\|=\|R_x^{u+K}(x_v)-x_u\|$ and $\forall x'\in X$, if
$x_v'\not\in\{x_v,R_x^{u+K}(x_v)\}$ then
$\|x_v-x_u\|\not=\|x'_v-x_u\|$.
\end{thm}
\begin{proof}
  Sketched in Fig.~\ref{f:prunedist} for $K=2$; the circles mark
  equidistant levels from 1. Intuitively, two branches from level 1 to
  level 4 or 5 will have equal segments but different angles, which
  will cause the end dots to be at different distances from level 1.
  The formal proof is by induction on the level distance. 
\end{proof}
\begin{figure}[!ht]
\begin{center}
\psfrag{nu1}{$\nu_1$}
\psfrag{nu2}{$\nu_2$}
\psfrag{1}{$1$}
\psfrag{2}{$2$}
\psfrag{r}{$5$}
\psfrag{g}{$3$}
\psfrag{b}{$4$}
\psfrag{3}{$\nu_3$}
\psfrag{4}{$\nu_4$}
\psfrag{5}{$\nu_5$}
\psfrag{6}{$\nu_6$}
\psfrag{7}{$\nu_7$}
\psfrag{8}{$\nu_8$}
\psfrag{9}{$\nu_9$}
\psfrag{10}{$\nu_{10}$}
\psfrag{11}{$\nu_{11}$}
\psfrag{12}{$\nu_{12}$}
\psfrag{13}{$\nu_{13}$}
\psfrag{14}{$\nu_{14}$}
\psfrag{15}{$\nu_{15}$}
\psfrag{16}{$\nu_{16}$}
\includegraphics[width=12cm]{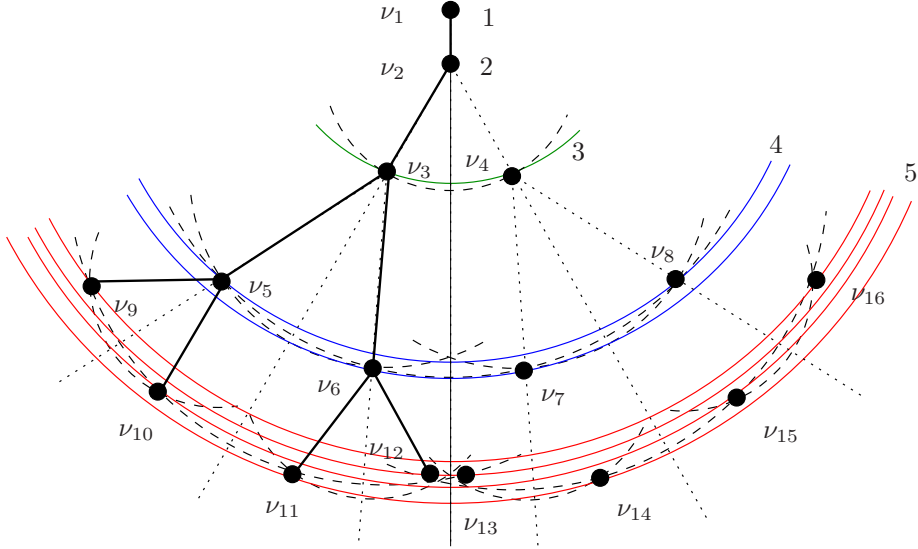}
\end{center}
\caption{A pruning edge $\{1,4\}$ prunes either $\nu_6,\nu_7$ or
  $\nu_5,\nu_8$.}
\label{f:prunedist}
\end{figure}

We now define partial reflection operators:
\begin{equation}
  g_v(x)=(x_1,\ldots,x_{v-1},R_x^v(x_v),\ldots,R_x^v(x_n)). \label{prefeq}
\end{equation}
The $g_v$'s map an embedding $x$ to its partial reflection with first
branch at $v$. It is evident that the $g_v$'s are injective with
probability 1 and idempotent.
\begin{lem}
  For $u,v\in V$ such that $u,v>K$, $g_ug_v(x)=g_vg_u(x)$.
  \label{lemmcomm}
\end{lem}
\begin{proof}
Assume without loss of generality $u<v$. Then:
\begin{eqnarray*}
  g_ug_v(x) &=& g_u(x_1,\ldots,x_{v-1},R_x^v(x_v),\ldots,R_x^v(x_n)) \\
            &=& (x_1\ldots,x_{u-1}, R_{g_v(x)}^u(x_u), \ldots,
                           R_{g_v(x)}^u R_x^v(x_v), \ldots, 
                           R_{g_v(x)}^u R_x^v(x_n)) \\
            &=& (x_1\ldots,x_{u-1}, R_{x}^u(x_u), \ldots,
                           R_{g_u(x)}^v R_x^u(x_v), \ldots, 
                           R_{g_u(x)}^v R_x^u(x_n)) \\
            &=& g_v(x_1,\ldots,x_{u-1},R_x^u(x_u),\ldots,R_x^u(x_n))
            \\
            &=& g_vg_u(x),
\end{eqnarray*}
where $R_{g_v(x)}^uR_x^v(x_w)=R_{g_u(x)}^vR_x^u(x_w)$ for each $w\ge
v$ by Lemma 4.2 in \cite{powerof2-tr}. 
\end{proof}
We define the {\it discretization group} to be the group
$\mathcal{G}_D=\langle g_v\;|\;v>K\rangle$ generated by the $g_v$'s
\begin{cor}
  With probability 1, $\mathcal{G}_D$ is an Abelian group isomorphic
  to $C_2^{n-K}$.
\end{cor}
For all $v>K$ let $\gamma_v=(1,\ldots,1,-1_v,\ldots,-1)$ be the vector
consisting of one's in the first $v-1$ components and $-1$ in the last
components. Then the $g_v$ actions are directly mapped onto the
chirality functions.
\begin{lem}
  For all $x\in X$, $\chi(g_v(x))=\chi(x)\odot \gamma_v$, where
  $\odot$ is the componentwise vector multiplication.
\end{lem}
\begin{proof}
This follows by definition of $g_v$ and of chirality of an embedding. 
\end{proof}

Because, by Alg.~\ref{alg:bp}, each $x\in X$ has a different
chirality, for all $x,x'\in X$ there is $g\in\mathcal{G}_D$ such that
$x'=g(x)$, i.e.~the action of $\mathcal{G}_D$ on $X$ is transitive. By
Thm.~\ref{pow2lemma}, the distances associated to the discretization
edges are invariant with respect to the discretization group.

\subsection{The pruning group}
Consider a pruning edge $\{u,v\}\in E_P$. By Thm.~\ref{pow2lemma},
with probability 1 we have $d_{uv}\in H^{uv}$, otherwise the instance
could not be a YES one. Also, again by Thm.~\ref{pow2lemma},
$d_{uv}=\|x_u-x_v\|\not=\|g_w(x)_u-g_w(x)_v\|$ for all
$w\in\{u+K,\ldots,v\}$ (note that distance $\|\nu_1-\nu_9\|$ in
Fig.~\ref{f:prunedist} is different from all its reflections
$\|\nu_1-\nu_h\|$ with $h\in\{10,11,13\}$ w.r.t.~$g_4,g_5$).  We
therefore define the {\it pruning group} $\mathcal{G}_P=\langle
g_w\;|\;w>K\land\forall \{u,v\}\in
E_P\;(w\not\in\{u+K,\ldots,v\})\rangle$. It is easy to show that
$\mathcal{G}_P\leq\mathcal{G}_D$. By definition, the distances
associated with the pruning edges are invariant with respect to
$\mathcal{G}_P$. 
\begin{thm}[Thm.~5.4 in \cite{powerof2-tr}]
The action of $\mathcal{G}_P$ on $X$ is transitive. \label{transthm}
\end{thm}
$|X|$ was shown to be a power of two with probability 1 in the
unpublished technical report \cite{powerof2-tr}. We provide an shorter
and clearer proof.
\begin{thm}
With probability 1, $\exists\ell\in\mathbb{N}\;|X|=2^\ell$. \label{pow2thm}
\end{thm}
\begin{proof}
  Since $\mathcal{G}_D\cong C_2^{n-K}$, $|\mathcal{G}_D|=2^{n-K}$.
  Since $\mathcal{G}_P\leq\mathcal{G}_D$, $|\mathcal{G}_P|$ divides
  the order of $|\mathcal{G}_D|$, which implies that there is an
  integer $\ell$ with $|\mathcal{G}_P|=2^\ell$. By
  Thm.~\ref{transthm}, the action of $\mathcal{G}_P$ on $X$ only has
  one orbit, i.e.~$\mathcal{G}_Px=X$ for any $x\in X$. By idempotency,
  for $g,g'\in\mathcal{G}_P$, if $gx=g'x$ then $g=g'$. This implies
  $|\mathcal{G}_Px|=|\mathcal{G}_P|$. Thus, for any $x\in X$,
  $|X|=|\mathcal{G}_Px|=|\mathcal{G}_P|=2^\ell$. 
\end{proof}

\subsection{The number of nodes in function of pruning edges}
Fig.~\ref{f:numsol} shows a Directed Acyclic Graph (DAG)
$\mathcal{D}_{uv}$ that we use to compute the number of valid nodes in
function of pruning edges between two vertices $u,v\in V$ such
that $v>K$ and $u<v-K$. 
\begin{figure}[!ht]
\hspace*{-1cm}\begin{minipage}{14cm}
\begin{center}
\psfrag{2}{1}
\psfrag{4}{2}
\psfrag{8}{4}
\psfrag{16}{8}
\psfrag{32}{16}
\psfrag{64}{32}
\psfrag{v}{$v$}
\psfrag{K+1}{\scriptsize $u\!\!+\!\!K\!\!-\!\!1$}
\psfrag{K+2}{\scriptsize $u\!\!+\!\!K$}
\psfrag{K+3}{\scriptsize $u\!\!+\!\!K\!\!+\!\!1$}
\psfrag{K+4}{\scriptsize $u\!\!+\!\!K\!\!+\!\!2$}
\psfrag{K+5}{\scriptsize $u\!\!+\!\!K\!\!+\!\!3$}
\psfrag{K+6}{\scriptsize $u\!\!+\!\!K\!\!+\!\!4$}
\psfrag{e1}{\scriptsize $0$}
\psfrag{e2}{\scriptsize $1$}
\psfrag{e3}{\scriptsize $2$}
\psfrag{e4}{\scriptsize $3$}
\psfrag{e5}{\scriptsize $4$}
\psfrag{e12}{\scriptsize $0\vee 1$}
\psfrag{e23}{\scriptsize $1\vee 2$}
\psfrag{e34}{\scriptsize $2\vee 3$}
\psfrag{e45}{\scriptsize $3\vee 4$}
\psfrag{e123}{\scriptsize $0\vee 1\vee 2$}
\psfrag{e234}{\scriptsize $1\vee 2\vee 3$}
\psfrag{e345}{\scriptsize $2\!\vee\!3\!\vee\!4$}
\psfrag{e1234}{\scriptsize $0\vee\ldots\vee 3$}
\psfrag{e2345}{\scriptsize $1\vee\ldots\vee 4$}
\psfrag{e12345}{\scriptsize $0\vee\ldots\vee 4$}
\includegraphics[width=14cm]{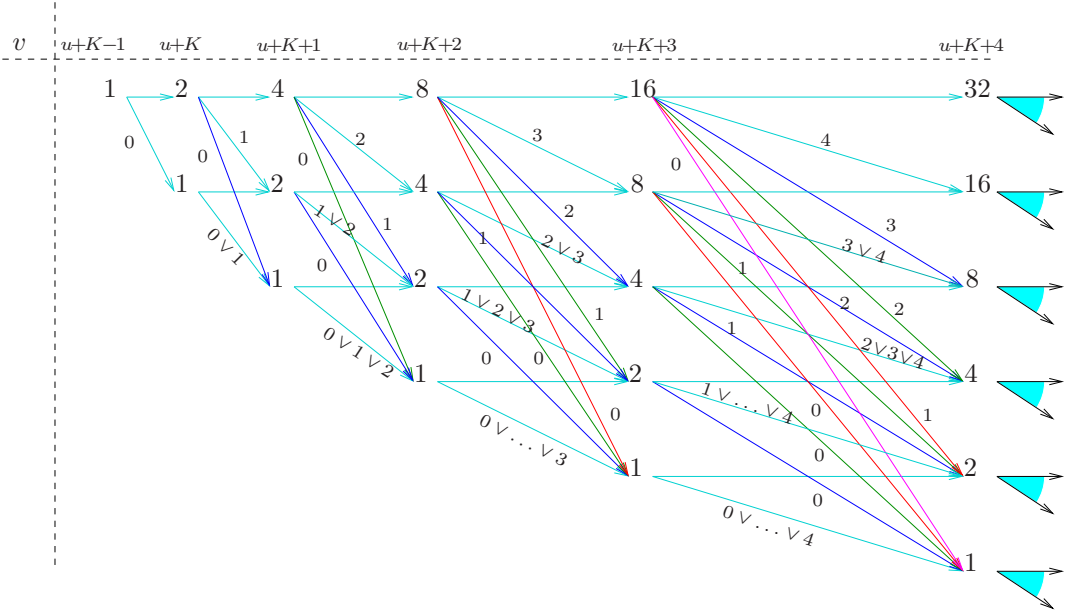}
\end{center}
\end{minipage}
\caption{Number of valid BP nodes (vertex label) at level $u+K+\ell$
  (column) in function of the pruning edges (path spanning all
  columns).}
\label{f:numsol}
\end{figure}
The first line shows different values for the rank of $v$ w.r.t.~$u$;
an arc labelled with an integer $i$ implies the existence of a pruning
edge $\{u+i,v\}$ (arcs with $\vee$-expressions replace parallel arcs
with different labels). An arc is unlabelled if there is no pruning
edge $\{w,v\}$ for any $w\in\{u,\ldots,v-K-1\}$. The vertices of the
DAG are arranged vertically by BP search tree level, and are labelled
with the number of BP nodes at a given level, which is always a power
of two by Thm.~\ref{pow2thm}. A path in this DAG represents the set of
pruning edges between $u$ and $v$, and its incident vertices show the
number of valid nodes at the corresponding levels. For example,
following unlabelled arcs corresponds to no pruning edge between $u$
and $v$ and leads to a full binary BP search tree with $2^{v-K}$ nodes
at level $v$.

\subsection{Polynomial DMDGP cases}
For a given $G_D$, each possible pruning edge set $E_P$ corresponds to
a path spanning all columns in $\mathcal{D}_{1n}$. Instances with
diagonal (Prop.~\ref{bwprop}) or below-diagonal (Prop.~\ref{bwprop2})
$E_P$ paths yield BP trees with constant width.
\begin{prop}
\label{bwprop}
If $\exists v_0>K$ s.t.~$\forall v>v_0$ $\exists!u<v-K$ with
$\{u,v\}\in E_P$ then the BP search tree width is bounded by
$2^{v_0-K}$.
\end{prop}
\begin{proof}
This corresponds to a path $\mbox{\sf
  p}_0=(1,2,\ldots,2^{v_0-K},\ldots,2^{v_0-K})$ that follows
unlabelled arcs up to level $v_0$ and then arcs labelled $v_0-K-1$,
$v_0-K-1\vee v_0-K$, and so on, leading to nodes that are all labelled
with $2^{v_0-K}$ (Fig.~\ref{f:p0}, left). 
\end{proof}
\begin{prop}
  \label{bwprop2}
  If $\exists v_0>K$ such that every subsequence $s$ of consecutive
  vertices $>\!\!v_0$ with no incident pruning edge is preceded by a
  vertex $v_s$ such that $\exists u_s<v_s\;(v_s-u_s\ge|s|\land
  \{u_s,v_s\}\in E_P)$, then the BP search tree width is bounded by
  $2^{v_0-K}$.
\end{prop}
\begin{proof}
  (Sketch) This situation corresponds to a below-diagonal path,
  Fig.~\ref{f:p0} (right). 
\end{proof}
\begin{figure}[!ht]
\hspace*{-1cm}\begin{minipage}{14.2cm}
\begin{center}
\psfrag{2}{\tiny 1}
\psfrag{4}{\tiny 2}
\psfrag{8}{\tiny 4}
\psfrag{16}{\tiny 8}
\psfrag{32}{\tiny 16}
\psfrag{64}{\tiny 32}
\psfrag{v}{}
\psfrag{K+1}{}
\psfrag{K+2}{}
\psfrag{K+3}{}
\psfrag{K+4}{}
\psfrag{K+5}{}
\psfrag{K+6}{}
\psfrag{e1}{}
\psfrag{e2}{}
\psfrag{e3}{}
\psfrag{e4}{}
\psfrag{e5}{}
\psfrag{e12}{}
\psfrag{e23}{}
\psfrag{e34}{}
\psfrag{e45}{}
\psfrag{e123}{}
\psfrag{e234}{}
\psfrag{e345}{}
\psfrag{e1234}{}
\psfrag{e2345}{}
\psfrag{e12345}{}
\includegraphics[width=7cm]{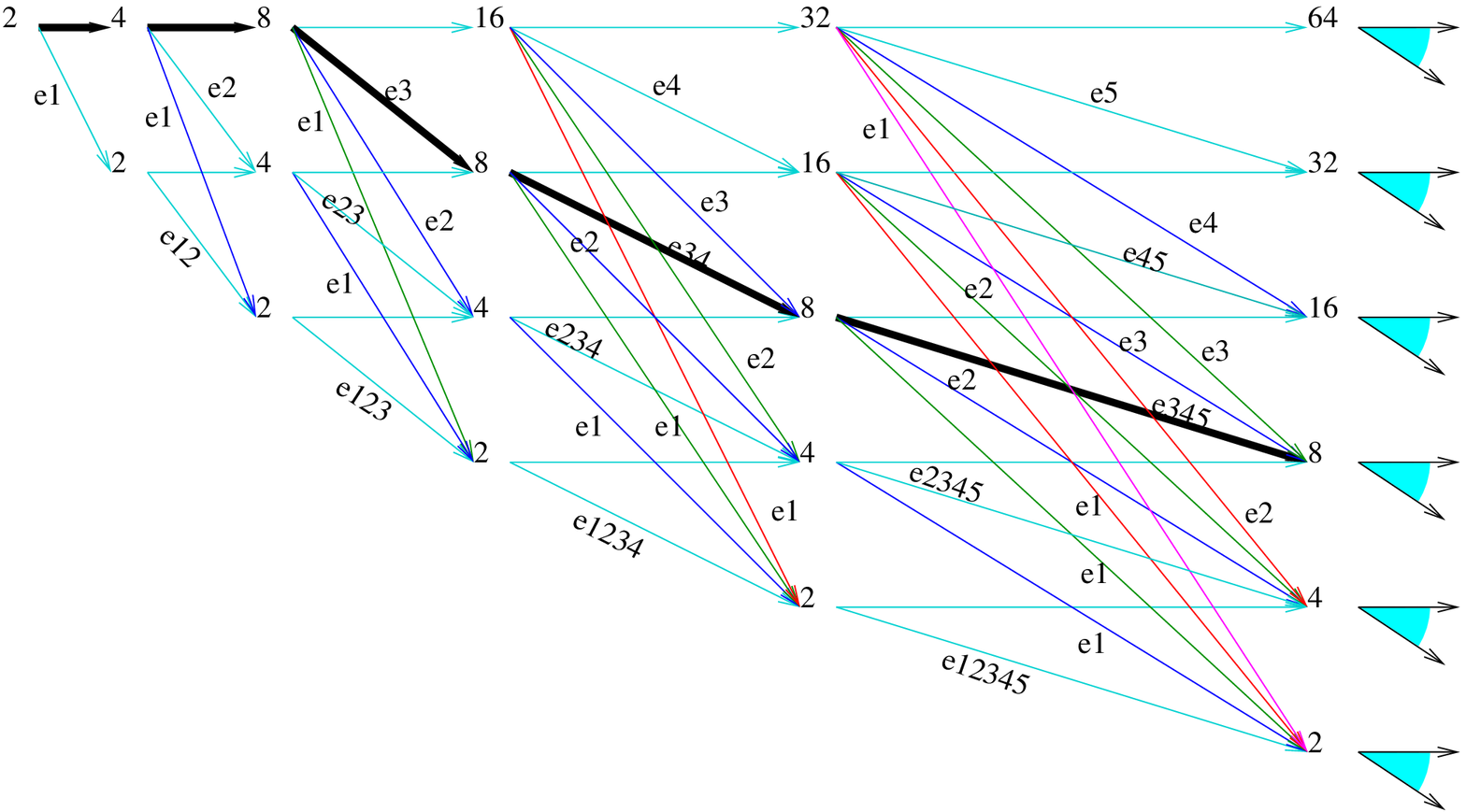}
\includegraphics[width=7cm]{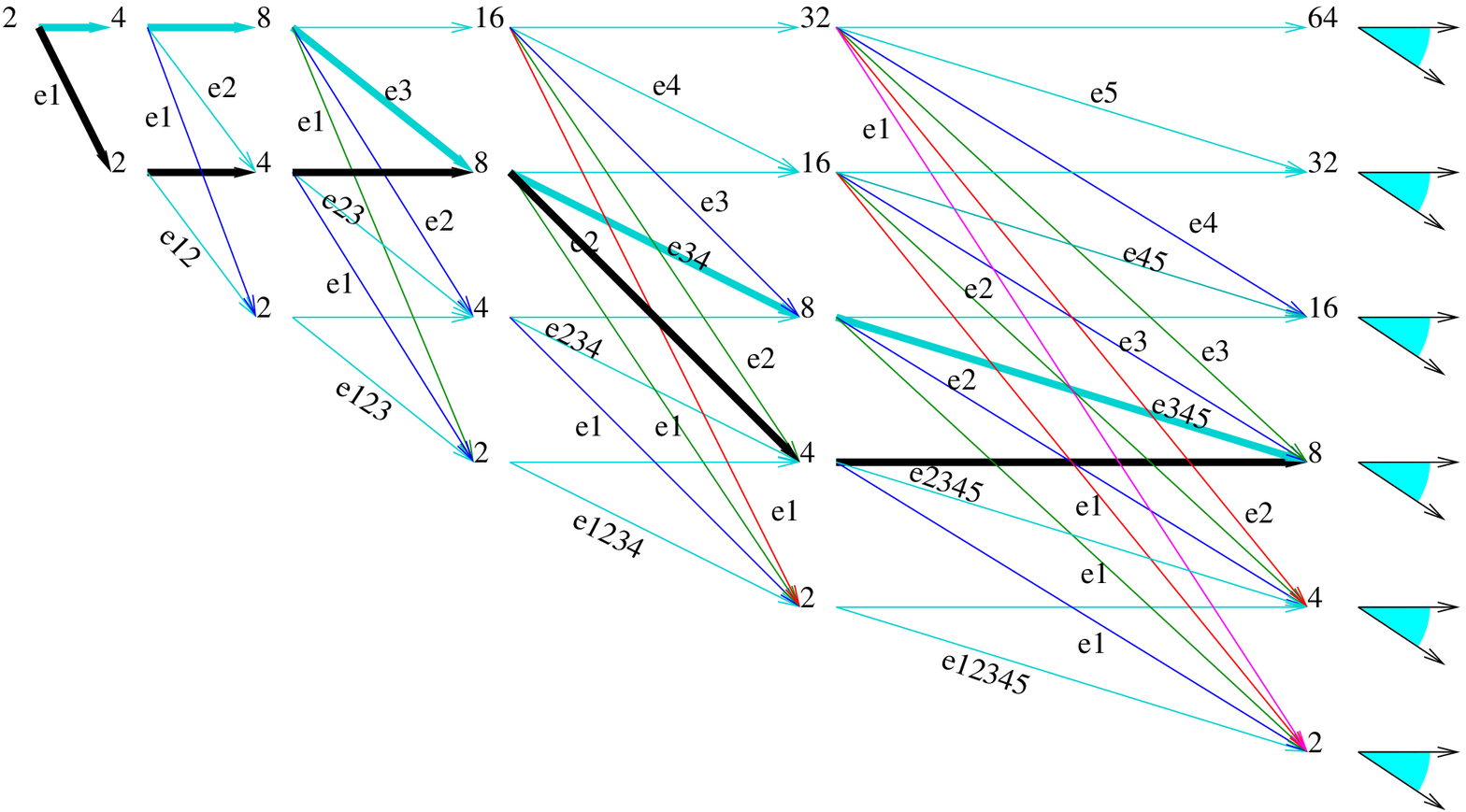}
\end{center}
\end{minipage}
\caption{A path $\mbox{\sf p}_0$ with treewidth $8$ (left) and another
  path below $\mbox{\sf p}_0$ (right).}
\label{f:p0}
\end{figure}
In general, for those instances for which the BP search tree width has
a $O(\log n)$ bound, the BP has a polynomial worst-case running time
$O(L 2^{\log n})=O(Ln)$, where $L$ is the complexity of computing $P$.
Since $L$ is typically constant in $n$ \cite{dong03}, for such cases
the BP runs in linear time $O(n)$.

Let $V'=\{v\in V\;|\;\exists \ell\in\mathbb{N}\;(v=2^\ell)\}$. 
\begin{prop}
  If $\exists v_0>K$ s.t.~for all $v\in V\smallsetminus V'$ with
  $v>v_0$ there is $u<v-K$ with $\{u,v\}\in E_P$ then the BP search
  tree width at level $n$ is bounded by $2^{v_0}n$. \label{bwprop3}
\end{prop}
\begin{proof}
  This corresponds to a path along the diagonal $2^{v_0}$ apart from
  logarithmically many vertices in $V$ (those in $V'$), at which
  levels the BP doubles the number of search nodes
  (Fig.~\ref{f:plog}).
\end{proof}
\begin{figure}[!ht]
\vspace*{-0.7cm}
\begin{center}
\includegraphics[width=7cm]{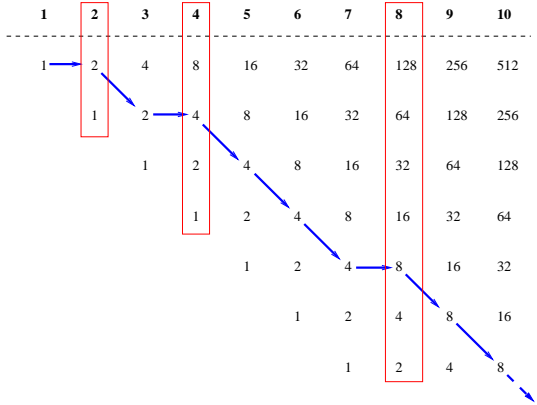}
\vspace*{-0.5cm}
\end{center}
\caption{A path with treewidth $O(n)$.}
\label{f:plog}
\end{figure}
For a pruning edge set $E_P$ as in Prop.~\ref{bwprop3}, or yielding a
path below it, the BP runs in quadratic time $O(n(n+1)/2)=O(n^2)$.

\subsection{Empirical verification}
On a set of sixteen protein instances from the Protein Data Bank
(PDB), twelve satisfy Prop.~\ref{bwprop}, and four
Prop.~\ref{bwprop2}, all with $v_0=4$. This is consistent with the
computational insight \cite{dmdgp} that BP has polynomial complexity
on real proteins.

\section{Conclusion}
We exploit some geometrical properties of an {\bf NP}-hard distance
geometry problem with a specific vertex order to derive some
polynomial cases. Empirically, proteins backbones seem to fall in
these cases; this provides an explanation for the practical efficiency
of a well-known embedding algorithm called Branch-and-Prune.

\bibliographystyle{plain}
\bibliography{pcc}

\end{document}